\documentclass[a4paper,6pt]{article}

\usepackage{booktabs,array,multirow}
\usepackage{amssymb}
\usepackage{amsmath}
\usepackage{amsthm}
\usepackage{relsize}
\usepackage{rotating}
\usepackage{amsfonts}
\usepackage{dsfont}
\usepackage{graphicx}
\usepackage{newlfont}
\usepackage{mathrsfs}
\usepackage{makecell}
\usepackage{makeidx}
\usepackage[utf8]{inputenc}
\usepackage{rotating}
\usepackage{txfonts}
\usepackage{mathptmx}
\usepackage{mathtools}
\usepackage{bbm}
\usepackage{float}
\usepackage{enumerate}
\usepackage{chngcntr}    
\usepackage{geometry}
\usepackage{enumitem}
\usepackage{xcolor}
\usepackage{authblk}
\usepackage{adjustbox}
\usepackage{tabularx}
\usepackage{adjustbox}
\usepackage[backend=bibtex, natbib=true, style=numeric, sorting=ynt]{biblatex}
\usepackage{etoolbox}
\usepackage{hyperref}

\bibliography{citations.bib} 

\newtheorem{lem}{Lemma}

\newtheorem{rem}{Remark}[section]

\newcommand{\into}{\rightarrow}
\newcommand{\E}{\mathbb{E}}

\newcommand{\PP}{\mathbb{P}}
\newcommand{\Real}{\mathbb{R}}

\newcommand{\F}{\mathpzc{F}}

\newcommand{\half}{\frac{1}{2}} 
\newcommand{\pder}[2]{\frac{\partial #1}{\partial #2}}
\newcommand{\qder}[2]{\frac{\partial^2 #1}{\partial^2 #2}}
\newcommand{\ppder}[3]{\frac{\partial^2 #1}{\partial #2 \partial #3}}
\newcommand{\tr}{\mathbf{tr}}

\DeclareMathAlphabet{\mathpzc}{OT1}{pzc}{m}{it}
\numberwithin{equation}{section}

\title{\textbf{Denoised Monte Carlo \\ for option pricing and Greeks estimation}}

\author[1]{Muchorski R. \thanks{Independent researcher, Ph.D. (Jagiellonian University)}}
\author[2]{Daniluk A.   \thanks{Independent researcher, Ph.D. (Jagiellonian University)}}
\author[1]{Lakshtanov E. \thanks{CIDMA, Depatment of Mathematics, University of Aveiro, 3810, Portugal}}

\affil[1]{\small\url{rafal.muchorski@gmail.com}}
\affil[2]{\small\url{andrzej.daniluk@gmail.com}}
\affil[3]{\small\url{lakshtanov@ua.pt}}

\begin{document}
\maketitle

\begin{abstract}
We present a novel technique of Monte Carlo error reduction that finds direct application in option pricing and Greeks estimation. The method is applicable to any LSV modelling framework and concerns a broad class of payoffs, including path-dependent and multi-asset cases. Most importantly, it allows to reduce the Monte Carlo error even by an order of magnitude, which is shown in several numerical examples.
\end{abstract}

\section{Introduction}
With only a few exceptions, including specific payoffs or underlying asset dynamics, option pricing requires the use of numerical techniques. Most often the only viable choice is a Monte Carlo (MC) simulation. This, however, poses the problem of the random error inherent to this method and involves a trade-off between the pricing accuracy and the computational cost. In order to deal with such problem, several techniques leading to some error reduction have been invented. The majority of those are based on concepts such as control variate (e.g. \cite{HullWhite}, \cite{DuLiuGu}, \cite{FouqHan07}) or importance sampling (e.g. \cite{GlasHei}, \cite{CapIS}, \cite{ZhaoLiu}). However, in many cases the applications of the known methods concern quite specific situations, pertaining either to the type of the payoff or asset dynamics.\\

In this paper we present a novel approach that is applicable to a broad class of option payoffs with a European exercise style (possibly multi-asset and path-dependent) and underlying processes (arbitrary pure diffusion, in general). \\ \\
The idea could be sketched as follows: 
\begin{enumerate}
	\item Given the dynamics of the underlying process, we introduce another, auxiliary process with simplified dynamics (typically an arithmetic or geometric Brownian motion). The only requirement is that the payoff and simplified dynamics admits a fast and easy calculation of the option price as a function of the underlying.
	\item Then we consider the option pricing function as if the dynamics of the underlying followed this auxiliary process.
	\item By setting the original underlying process as an argument to this pricing function, we construct a process whose terminal value is the option payoff. At the same time, we explicitly decompose this process into two parts: a martingale and an integral of a drift, which can be calculated explicitly. 
\end{enumerate}
The crux of the method is that, for the purpose of calculating the expected value, the martingale part of the latter process can be dropped, and only the drift integral is to be calculated. However, this component typically accounts for only a tiny part of the randomness of the whole process. Hence, if it is calculated using MC simulation, its error is much smaller than that of the option payoff itself. \\ \\
The paper is organized as follows. In the first section, we present a formal description of the proposed approach, together with the main theoretical result. First, we do it for a European payoff and then show how it could be generalized to a path-dependent case. In the second section, we present numerical results obtained for various option payoffs (including path-dependent and multi-asset options) for selected stochastic volatility models (Heston and SABR). We compare the simulation prices and errors from our method with the corresponding results obtained in crude MC simulations.

\section{Main result}
Let $\left(\Omega,\F,\PP\right)$ be a filtered probability space with filtration $\ F = \left(\F_t\right)_{t \in [0,T]}, T > 0 $ 
and let $W$, $\tilde{W}$ be standard R-dimensional ($R \geq 1$), $\F$ -adapted Wiener processes under $\PP$. 
For a given random variable $Z$ on $\Omega$ we denote $\E_t Z \coloneqq \E^\PP \left( Z|\F_t \right)$. 
Henceforth, we will also assume that all expected values which appear in formulas exist. 
\\
Consider a pair of continuous, square-integrable, d-dimensional ($R \leq d$) diffusion processes 
\begin{equation}
X, \tilde{X}: [0, T] \times \Omega \into \mathcal{C} \subset \Real^d
\end{equation}
where $ \mathcal{C} = (\underline{B}_1, \overline{B}_1) \times \cdots \times (\underline{B}_d, \overline{B}_d)$ 
for some $ -\infty \leqslant \underline{B}_k < \overline{B}_k \leqslant +\infty, k = 1, \ldots, d $. 
\\ \\
Moreover, we assume that $\tilde{X}$ attains each its point of $\mathcal{C}$ , i.e. for any $t \in [0,T], x \in \mathcal{C} $ the density of $\tilde{X}_t$ at $x$ is positive. Suppose that $X, \tilde{X}$ are solutions of the SDEs
\begin{equation}\label{dyn_orig}
dX_t = \mu_t dt + \sigma_t dW_t
\end{equation}
\begin{equation}\label{dyn} 
d \tilde{X}_t = \tilde{\mu} \left(t, \tilde{X}_t \right) dt + \tilde{\sigma} \left(t, \tilde{X}_t \right) d\tilde{W}_t
\end{equation}
where 
\begin{equation}
\sigma: [0, T] \times \Omega \into \mathcal{M}(d,R), \quad \tilde{\sigma}: [0, T]  \times \mathcal{C}^d \into \mathcal{M}(d,R)
\end{equation}
are respectively some stochastic volatility process and a local volatility function, where $\mathcal{M}(d,R)$ is the space of $d \times R$ matrices. Similarly, $\mu$ and $\tilde{\mu}$ are stochastic drift vectors. 

\subsection{European payoff}
Let's consider a square-integrable random variable $Z$ of the form $Z = \pi(X_T)$, where $ \pi: \mathcal{C}^d \into \Real $ is a given Borel function. 
Suppose that we are interested in calculating expected value of $Z$. For this purpose, let us introduce an auxiliary variable 
$ \tilde{Z} \coloneqq \pi(\tilde{X}_T) $ and consider its conditional expectation $ \E_t \tilde{Z} $.
Since the process $\tilde{X}$ is Markovian, it holds $ \E_t \tilde{Z} = \E^\PP \left( \tilde{Z} | X_t \right) $, 
so there exists a function $ \psi: [0,T] \times \mathcal{C} \into \Real $, such that 
\begin{equation}\label{eq:psi}	
\E_t \tilde{Z}  = \psi \left( t, \tilde{X}_t \right)
\end{equation}
Moreover, by the Feynman-Kac theorem $\psi$ is a solution of the PDE
\begin{equation}\label{pde}
\pder{\psi(t,x)}{t} + \delta(t,x)^T \tilde{\mu}(t,x) + \half \tr \left( \tilde{\sigma}(t,x)^T \gamma(t,x) \tilde{\sigma}(t,x) \right) = 0 
\end{equation}
where
\begin{equation}
\delta(t,x) \coloneqq \left[ \pder{\psi(t,x)}{x_k} \right]_{k=1,\ldots,d}, \quad \gamma(t,x) \coloneqq \left[ \ppder{\psi(t,x)}{x_j}{x_k} \right]_{j,k=1,\ldots,d}
\end{equation}

Now, let us define the process 
\begin{equation}
V_t = \psi \left( t, X_t \right)
\end{equation}
Note that by definition, for any value of $\tilde{X}_T$ 
\begin{equation}
\psi \left( T, \tilde{X}_T \right) = \E_T \pi(\tilde{X}_T) = \pi(\tilde{X}_T)
\end{equation}
so by virtue of our assumption $ \psi \left( T, x \right) = \pi(x) $ for any $ x \in  \mathcal{C} $. In particular
\begin{equation}
V_T = \psi \left( T, X_T \right) = \pi(X_T) = Z
\end{equation}
Furthermore, from the Ito formula it follows	
\begin{equation}
dV_t = \pder{\psi}{t}(t,X_t) dt + \delta(t,X_t)^T \left( \mu_t dt + \sigma_t dW \right) + \half \tr \left( \sigma_t^T \gamma(t,X_t) \sigma_t \right) dt
\end{equation}

Setting $x=X_t$ in the equation \ref{pde} and subtracting its left-hand side from the formula above we get
\begin{equation}
dV_t = \xi_t dt + \delta(t,X_t)^T \sigma_t dW_t
\end{equation}
where
\begin{equation}\label{defxi}
	\xi_t = \delta(t,X_t)^T \left(\mu_t - \tilde{\mu}(t,X_t) \right) +
	\half \tr \left( \sigma_t^T \gamma(t,X_t) \sigma_t - \tilde{\sigma}(t,X_t)^T \gamma(t,x) \tilde{\sigma}(t,X_t) \right)
\end{equation}
In an integral form
\begin{equation}
V_T = V_0 + \int_0^T \xi_t dt + \int_0^T \delta(t,X_t)^T \sigma_t dW_t
\end{equation}
Taking expected value we get
\begin{equation}
\E V_T = V_0 + \E \int_0^T \xi_t dt
\end{equation}
or equivalently 
\begin{equation}\label{europ_result}
\E Z = \psi \left( 0, X_0 \right) + \E \int_{0}^{T} \xi_t dt
\end{equation}
which is our main result. \\

\begin{rem} 
	The formula \ref{europ_result} can be viewed as a calculation of the expected value assuming the "wrong", simplified dynamics of the underlying process, plus a correction term for this dynamics change. 
\end{rem}	
\begin{rem} 
	The formula \ref{europ_result} is directly applicable to option pricing under the $T$-forward measure, in which the expected payoff is undiscounted. However, it is not restrictive to this very measure. Indeed, any numéraire with dynamics fitting the general form \ref{dyn} can be included as an additional $d+1$'th component of the processes $X$ and $\tilde{X}$. Then, with a little modification of the function $\pi$, the variable $Z$ can include the numéraire in the denominator, while still being a function of $X_T$.		
\end{rem}		

\subsection{Path-dependent payoffs}
In general, a path-dependent payoff takes the form $ \pi^{*} \left(\left(X_s\right)_{s \in [0,T]} \right) $, where $\pi^*$ is a functional determined for any continuous path of observed values of the process $X_t$. However, given that in practice both observed values and time are discrete, in order to simplify the formalism, we will consider payoffs of the form 
	$Z = \pi \left( X_{T_0}, \ldots, X_{T_n} \right)$, where at discrete moments $0 = T_0 < T_1 < \ldots < T_n \leqslant T$. 
\\ \\
Similar as in the European case, we introduce analogous variable 
	$ \tilde{Z} = \pi \left( \tilde{X}_{T_1}, \ldots, \tilde{X}_{T_n} \right) $ 
and consider its conditional expectation $ \E_t \tilde{Z} $. 
Then, due to the Markovian property of $\tilde{X}$
\begin{equation}
\E_t \tilde{Z} = \E^\PP \left( \tilde{Z}  | \tilde{X}_t, \left\{ \tilde{X}_{T_j}: T_j < t \right\} \right)
				  = \E^\PP \left( \tilde{Z}  | \tilde{X}_t, \left\{ \tilde{X}_{T_j}: T_j \leqslant t \right\} \right)
\end{equation}				  
So, there exist functions 
	$ \psi_1, \ldots, \psi_n, \; \psi_k: [T_{k-1}, T_k] \times \mathcal{C}^d \times (\mathcal{C}^d)^{k-1} \into \Real $, 
such that for $ T_{k-1} \leqslant t \leqslant T_k $ it holds
\begin{equation}
\E_t \tilde{Z} = \psi_k \left( t, \tilde{X}_t; \tilde{\mathbf{Y}}_k \right)
\end{equation}
where $ \tilde{\mathbf{Y}}_k = \left[ \tilde{X}_{T_j} \right]_{1 \leqslant j < k} $ 
is a vector of values of $\tilde{X} $ actually observed prior to the time $t$ and can be viewed as parameter of the function $\psi_k$ (in the degenerate case $k = 0$ the vector $\tilde{\mathbf{Y}}_k$ is not defined and shall be skipped). 
Notice that 
\begin{equation}\label{eq:psi_k}
\psi_k \left( T_k, \tilde{X}_{T_k}; \tilde{\mathbf{Y}}_k \right) = \psi_{k+1} \left( T_k, \tilde{X}_{T_k}; \tilde{\mathbf{Y}}_{k+1} \right)
\end{equation}
\\ 
Next, we define the process $V_t$ so that
\begin{equation}
V_t \coloneqq \psi_{k}\left(t, X_t; \mathbf{Y}_k\right), \quad  t \in [T_{k-1},T_k]
\end{equation}
where $ \mathbf{Y}_k = \left[X_{T_j}\right]_{1 \leqslant j < k}$ is a vector of discrete values of $X_t$, observed until the time $T_k$ 
\\\\
Mimicking the reasoning presented in the previous subsection and defining analogously
\begin{equation} 
\delta_k (t,x;y) \coloneqq \left[ \pder{\psi_k (t,x;y)}{x_j} \right]_{j=1,\ldots,d}, \quad
\gamma_k (t,x;y) \coloneqq \left[ \ppder{\psi_k (t,x;y)}{x_i}{x_j} \right]_{i,j=1,\ldots,d}
\end{equation}
\begin{equation}\label{defxi2}	   
	\xi_{t,k} = \delta_k (t,X_t;\mathbf{Y}_k)^T \left(\mu_t - \tilde{\mu}(t,X_t) \right) +
    \half \tr \left( \sigma_t^T \gamma_k (t,X_t;\mathbf{Y}_k) \sigma_t - \tilde{\sigma}(t,X_t)^T \gamma_k (t,X_t;\mathbf{Y}_k) \tilde{\sigma}(t,X_t) \right) 
\end{equation}    
\vspace{2mm} 
we find that 
\begin{equation}
\E_{T_{k-1}} V_{T_k} = V_{T_{k-1}} + \E_{T_{k-1}} \int_{T_{k-1}}^{T_k} \xi_{t,k} dt
\end{equation}
\vspace{2mm} 
By applying the expected value to both sides and using the tower property, we obtain
\begin{equation}\label{increment_k}
	\E V_{T_k} - \E V_{T_{k-1}} = \E \int_{T_{k-1}}^{T_k} \xi_{t,k} dt
\end{equation}
Finally, we notice that by definition
\begin{equation}
	V_{T_n} = \psi_n\left(T,X_T;\mathbf{Y}_n\right) = \pi_n\left(\mathbf{Y}_n,X_T\right) = \pi\left(X_{T_0}, \ldots, X_{T_n}\right) = Z
\end{equation}
Hence, summing up \eqref{increment_k} over $k$ we get
\begin{equation}\label{path_dep_result}
	\E Z = \psi_1\left(0, X_0\right) + \E \sum_{k=1}^{n} \int_{T_{k-1}}^{T_k} \xi_{t,k} dt = \psi_1 \left( 0, X_0 \right) + \E \int_{0}^{T} \xi_t dt
\end{equation}
where we denoted $\xi_t \coloneqq \xi_{t,k(t)}, \; k(t) \coloneqq min \{k: t \le T_k \}$
\\
\begin{rem}
Although we considered only a discrete observation case, the formula for $\E Z$ does not assume any specific frequency of observations, which can be arbitrarily dense over time. This allows to use our method to approximate the value of the payoff also if observations are done in continuous time. Hence, it would be tempting to extend the formula \ref{path_dep_result} and assume it to be valid also in for such payoffs. This is indeed possible in certain cases, but requires some caution. Namely, both functions $\psi$ and $\xi$ derived from it may converge to the corresponding continuous case with an increasing number of observations. However, there is no guarantee that the expected value of the integral in \ref{path_dep_result} will also converge. A detailed discussion of this problem goes far beyond the scope of this paper, though.
\end{rem}

\subsection{Greeks calculation}
Our approach can be used not only for pricing, but also for the calculation of Greek coefficients, such as delta and gamma. For reasons of brevity, we present the derivation only for the European (path independent) case and $d=1$. Path-dependent and/or multidimensional cases are analogous, however, the formulas become more complicated. 
\\ \\
Namely, let us notice that $\xi_t$ in \ref{defxi} has a form
\begin{equation}
\xi_t = \xi \left( t, X_t, \mu_t, \sigma_t \right)
\end{equation}
Thus, we can calculate option delta by differentiating \ref{europ_result} as
\begin{equation}
\Delta = \pder{}{X_0} \psi \left( 0, X_0 \right)  + \pder{}{X_0} \E \int_{0}^{T} \xi_t dt = 
				\pder{}{X_0} \psi \left( 0, X_0 \right)  + \E \int_{0}^{T} \pder{}{X_0} \xi \left( t, X_t, \mu_t, \sigma_t \right) dt
\end{equation}
The $1^{st}$ term in the equation above is nothing but the option delta in the simplified model, which can be calculated easily an precisely by choosing tractable simplified dynamics. In turn, the integrand in the $2^{nd}$ term can be expressed as
\begin{equation}
\pder{}{X_0} \xi \left( t, X_t, \mu_t, \sigma_t \right) = \pder{X_t}{X_0} \pder{\xi}{X} \left( t, X_t, \mu_t, \sigma_t \right) + 
		\pder{\mu_t}{X_0} \pder{\xi}{\mu} \left( t, X_t, \mu_t, \sigma_t \right) + \pder{\sigma_t}{X_0} \pder{\xi}{\sigma} \left( t, X_t, \mu_t, \sigma_t \right)
\end{equation}		
where the derivatives of processes $X_t, \mu_t, \sigma_t $ in respect to $X_0$ are interpreted as path-wise. 
\\ \\
In general, these derivatives can be calculated numerically, either by using finite difference scheme or using Algorithmic Differentiation technique (see e.g. \cite{BBCD}, \cite{CapAAD}, \cite{Henr}, \cite{Savin}). 
Note that our method is well suited for the Adjoint Differentiation method, as the latter requires work memory  proportional to the number of operation of the feeding algorithm, so the possibility to compute the impact of each path independently is essential.
\\ \\
It is also worth mentioning that in a special, yet important case, no numerical differentiation technique is needed. Namely, if $X$ follows the dynamics with the solution of the general form
\begin{equation}
X_t = X_0 \exp \left( \int_{0}^{t} \sigma_u dW_u + \half \int_{0}^{t} \left(\mu(u) - \sigma^2_u\right) du \right)
\end{equation}
where $\mu(t) = \frac{\tilde{\mu}(t,x)}{x}$ and $\sigma_u$ is not explicitly dependent on $X_t$ , then 
\begin{equation}
\pder{X_t}{X_0} = \frac{X_t}{X_0}
\end{equation}
and finally 
\begin{equation}
\Delta = \pder{}{X_0} \psi \left( 0, X_0 \right) + \E \int_{0}^{T} \frac{X_t}{X_0} \pder{\xi}{X} \left( t, X_t, \mu_t, \sigma_t \right) dt
\end{equation}
Applying the same reasoning we find 
\begin{equation}
\Gamma = \qder{}{X_0} \psi \left( 0, X_0 \right) + \E \int_{0}^{T} \frac{X^2_t}{X^2_0} \qder{\xi}{X} \left( t, X_t, \mu_t, \sigma_t \right) dt
\end{equation}

\subsection{Application to Monte Carlo}
We derived formulas that provide a practical and efficient way of calculating $\E Z$. Indeed, the function $\psi$ can be calculated quite easily and precisely by solving the PDE \ref{pde}, either analytically or numerically. In turn, the expected value of the "correction term" $ \int_{0}^{T} \xi_t dt $ can be computed by means of an MC simulation. Although this is still associated with a sampling error, such calculation has a significant advantage compared to a generic MC simulation of $Z$. This is due to the fact that, in practical cases, the correction term is small, and its volatility is much smaller than the volatility of $V_T$. 
\\ \\
This benefit does not come without a cost, as it requires path-wise evaluation of $\xi_t$. This additional computational cost is not large, though, when we realize that $\xi_t$ does not need to be evaluated at each simulation time step on the path. Indeed, what we need to calculate is only
\begin{equation}
 J \coloneqq \E \int_{0}^{T} \xi_t dt = \int_{0}^{T} \overline{\xi} (t) dt
\end{equation}
where $ \overline{\xi} (t) \coloneqq \E \xi_t $ is a continuous function of time. 
Therefore, usually the integral can be very well approximated with the use of quadratures. Namely 
\begin{equation}
J \approx T \sum_{k=1}^{L} w_k \overline{\xi} (a_k T)
\end{equation}
where $ \left[w_k\right]_{k=1,\ldots,L}, \; \left[a_k\right]_{k=1,\ldots,L} $ 
are weights and abscisas of Gauss-Legendre quadrature, respectively. 
Given that $ \E \xi_t $ can be approximated by an average over simulated paths
\begin{equation}
\E \xi_t \approx \frac{1}{N} \sum_{k=1}^{N} \xi_t (\omega_n)
\end{equation}
where $\omega_n$ identify an individual scenario and $N$ is the number of MC simulations, we finally get
\begin{equation}
J \approx \frac{T}{N} \sum_{n=1}^{N} \sum_{k=1}^{L} w_k \xi_{a_k T}(\omega_n)
\end{equation}
\\ \\
Another optimization is possible for payoffs including multiple underlyings ($d>1$) and assuming a special form of the system of SDEs explaining the asset dynamics in \eqref{dyn_orig}-\eqref{dyn}. Namely, if each asset is governed by its own SDE, with Wiener processes that are correlated for a given constant correlation matrix $\rho=[\rho_{ij}]_{i,j=1 \ldots d}$ with rank R ($1 \leq R \leq d$), so that

\begin{equation}
d \langle W^{i}, W^{j} \rangle_t =  d \langle \tilde{W}^{i}, \tilde{W}^{j} \rangle_t = \rho_{ij} dt \quad 1 \leq i,j \leq R
\end{equation} 

then \eqref{dyn_orig}-\eqref{dyn} can be expressed as

\begin{equation}\label{dyn_orig_simple}
dX_t = \mu_t dt + diag(\sigma_t) C dW_t
\end{equation}
\begin{equation}\label{dyn_simple} 
d \tilde{X}_t = \tilde{\mu} \left(t, \tilde{X}_t \right) dt + diag(\tilde{\sigma} \left(t, \tilde{X}_t \right)) C d\tilde{W}_t
\end{equation}

with:
\begin{itemize}
\item[$\bullet$] $d \times R$ matrix $C=[c_{jk}]_{j=1\ldots d,k=1\ldots R}$, such that $CC^T=\rho$, and
\item[$\bullet$] diagonal matrices $diag(\sigma_t),diag(\tilde{\sigma} \left(t, \tilde{X}_t \right))$ with, respectively, volatilities $\sigma_t,\tilde{\sigma} \left(t, \tilde{X}_t \right)$ on the diagonal.
\end{itemize}

With such simplified dynamics, we can further optimize the calculation of $\xi_t$ from \ref{defxi} in the European payoff case, and also $\xi_{t,k}$ from \ref{defxi2}, in the path-dependent case. At the first glance it may seem that in both cases the associated computational cost is $O(d^2)$, since the formulas involve $d \times d$ Hessian matrices $\gamma,\gamma_k$. If each Hessian element was to be calculated using a finite difference method, this would require $d^2 + d + 1$ calls of the pricing function $\psi$. However, we don't really need to calculate the whole Hessian, but only some related quadratic forms, with their alternative expressions, namely
\begin{equation}\label{eq:first set}
\tr \left( \sigma_t^T \gamma(t,x) \sigma_t \right), \quad \tr \left( \tilde{\sigma}(t,x)^T \gamma(t,x) \tilde{\sigma}(t,x) \right)
\end{equation}
in the European payoff case, and
\begin{equation}\label{eq:second set}
\tr \left( \sigma_t^T \gamma_k (t,x;y) \sigma_t \right), \quad \tr \left( \tilde{\sigma}(t,x)^T \gamma_k (t,x;y) \tilde{\sigma}(t,x) \right)
\end{equation}
in the path-dependent payoff case. As we demonstrate below, this can be done in linear time $O(d)$.\\\\
We start from a simple lemma that explains the benefits of calculating particular Laplacian values, including a change of variables, to obtain the aforementioned quadratic forms

\begin{lem}\label{lemma}
Let $f$ be a real function that is twice differentiable in $x \in \Real^d$, and let $A = diag(\tilde{c}) C$, where $diag(\tilde{c})$ is a diagonal matrix with vector $\tilde{c}=[\tilde{c}_1,\ldots,\tilde{c}_d]$ on the diagonal, with positive entries. We define the linear mapping
\begin{equation}
\tau(u) \colon \Real^R \to \Real^d \quad \tau(u) \coloneqq A u^T
\end{equation}
where $u^T$ is the column vector of the arguments $u \in \Real^d$ of $\tau$. Then for
\begin{equation}
u(x) =\left(u_1(x),\ldots,u_R(x)\right) \colon \tau(u(x)) = x
\end{equation}
holds 
\begin{equation}\label{eq:lemm}
\nabla^{2} (f \circ \tau) (u(x)) = \sum_{i=1}^R \frac{\partial^2 (f \circ \tau)}{\partial u_i^2} (u(x)) = \tr \left( A^T \kappa(x) A \right)  
\end{equation}
where
\begin{equation}
\kappa(x) \coloneqq \left[ \ppder{f(x)}{x_j}{x_k} \right]_{j,k=1,\ldots,d}
\end{equation}
\end{lem}
\begin{proof}
The proof can be found in the Appendix.
\end{proof}

Now, let's notice that for fixed $t,x,y$ ($y$ for the path-dependent payoff case) and $\omega \in \Omega$ we can regard both $diag(\sigma_t) = diag(\sigma_t(\omega))$ and $diag(\tilde{\sigma}(t,x))$ as diagonal matrices with fixed, positive entries on the diagonal. Let's denote $\lambda_i, v_i$ as respectively the $i$-th (positive) eigenvalue and the corresponding eigenvector of $\rho$ ($1 \leq i \leq R$). Then for the matrix $C$ with $R$ column vectors $\sqrt{\lambda_i} v^T_i$ we can apply Lemma \ref{lemma} to cases $\tilde{c} = \sigma_t$ and $\tilde{c} = \tilde{\sigma}(t,x)$ to express the quadratic forms from \eqref{eq:first set} and \eqref{eq:second set} as certain Laplacian values, as in \eqref{eq:lemm}. Consequently, $\xi_t$ and $\xi_{t,k}$ can be calculated with the use of $2R$ directional derivatives. Each of them requires calculating finite differences with 2 additional evaluations of the functions $\psi,\psi_k$, and this can be done with $4R+1 \leqslant 4d+1 $ evaluations of each function. In addition, the matrix $C$ can be precalculated once for the known correlation matrix $\rho$ and this operation has only marginal contribution to the total computational cost. This type of optimization has been implemented for the purpose of numerical testing, described in the following Section \ref{numerical test}.

\section{Numerical tests}\label{numerical test}

\subsection{Test design}
In this section, we examine our approach empirically for various options option types and parameters. For this purpose we compare estimators of option values obtained from Monte Carlo simulations in two ways: i) as a crude payoff average and ii) with the use of our formulas. For our method we used a moderate number of $N = 5,000$ simulations. To calculate the integral term in \ref{europ_result} and \ref{path_dep_result}, we tested 2 formulas: a (left) Riemann sum with $\Delta t = 0.0001$ and a Gauss-Legendre quadrature with $L = 24$ nodes. As a benchmark we used the average of option payoff based on $N_{bmk} = 1,000,000$ paths. In addition, we calculated a standard error of our and crude MC estimators. For the sake of comparability this was calculated for the same number of $N$ simulations for both methods. 
Analogous calculations were performed for option deltas (except of Riemann sum in our method, which would be too demanding numerically). 
HPC computations for benchmarks and Differentiation was done using Automatic Adjoint Differentiation (AAD) technique for which we used a publicly available calculation tool\footnote{AADC Community Edition by MatLogica,  https://matlogica.com/pricing.php } 
\\ \\
We examined the following types of option payoffs:
\begin{table}[H]
\centering
	\begin{tabular}{c|l|l}
		\toprule
		\# & \multicolumn{1}{c}{Description} & \multicolumn{1}{c}{Payoff} \\
		\midrule
	 	 1 & Vanilla: European plain vanilla call           &  $Z = \max(0, X_T-K)$ \\
		 2 & Barrier: Down-and-out barrier call (no rebate) &  $Z = \max(0, X_T-K) \cdot 1_{m_T > H}, \quad m_T = \min_{0 \leqslant t \leqslant T} X_t $ \\
		 3 & Asian: Asian call (with quarterly averaging)   &  $Z = \max(0, X_{avg}-K), \quad \quad X_{avg} = \frac{1}{4T} \sum_{k=1}^{4T} X_{k/4} $   \\    
		 4 & Basket: Call on an basket (equally weighted)   &  $Z = \max(0, X_{avg}-K), \quad \quad X_{avg} = \frac{1}{10} \sum_{k=1}^{10} X_T^{(k)} $ \\
		 5 & Rainbow: Call on a maximum of 3 assets         &  $Z = \max(0, X_{max}-K), \quad \quad X_{max} = \max_{k=1,2,3} X_T^{(k)}$ \\
		\bottomrule
	\end{tabular}
 \end{table}
For each option type we considered 2 maturities: $T=1$ and $T=5$ and 2 strikes: at-the-money-forward (approximately) and out-of-the-money. Strikes of OTM options were set so that they have similar moneyness to allow for comparability. Exact values of ATMF and OTM strikes are provided in the table with calculation results. 
\\ \\
For each options type calculations were performed assuming 2 different underlying dynamics, described by Heston and SABR models and parameterized as follows:
\begin{table}[H]
	\centering
	\begin{tabular}{l|c|c}
		\toprule
		Name & Dynamics & Parameters \\
		\midrule
		Heston & \makecell{$ dX_t = X_t \left( r dt + \sqrt{v_t} dW_t \right) $ \\ $ dv_t = \kappa (\theta - v_t) dt + \gamma \sqrt{v_t} dZ_t $ } 
		       & \makecell{$ X_0 = 100, r = 0.05 $ \\ $ \theta = v_0, \kappa = 5, \gamma = 0.3, \rho = -0.1 $ }  \\
		\midrule
		SABR   & \makecell{ $ dX_t = v_t X_t^\beta  dW_t $ \\ $ dv_t = \alpha v_t dZ $ } 
			   & $\alpha = 0.4, \beta = 0.5, \rho = 0 $  \\
		\bottomrule
	\end{tabular}
\end{table}		
where $ dW_t dZ_t = \rho dt $. In terms of initial values, in all cases we took $X_0 = 100$. In respect of $v_0$ we differentiate between single-asset and multi-asset options. In case of Vanilla, Barrier and Asian options we took $v_0 = 0.01 $ for Heston and $\sigma_0 = 2.5$ for SABR, respectively. In case of basket options, we assumed each asset to have a different $v_0$ taking the values:
\begin{itemize}
	\item for Heston: $ v_0 = 0.0036, 0.0049, 0.0064, 0.0081, 0.01, 0.0121, 0.0144, 0.0169, 0.0196, 0.0225 $
	\item for SABR:   $ v_0 = 1.8, 2.0, 2.2, 2.4, 2.6, 2.8, 3.0, 3.2, 3.4, 3.6 $
\end{itemize}

In case of rainbow options we only tested the SABR model, assuming step size $\Delta t = 0.0002$, maturity $T=1$ and 3 assets having different $v_0$ values $v_0 = 2, 2.5, 3$, with equal correlation across all asset price pairs $\rho_{i,j} \equiv 0.4$ ($i \neq j$) and independent stochastic volatility processes.

As the simplified, tractable dynamics we used either Black-Scholes or Bachelier model, which in our formalism correspond to $\tilde{\sigma_t} = \tilde{\sigma} X_t^B $, where $B=1$ for Black-Scholes and $B=0$ for Bachelier. The values of $\tilde{\sigma}$ were chosen in accordance with the original dynamics, so asset volatilities at the inception coincide $\tilde{\sigma_0} = \sigma_0$, which is summarized in the following table:
\begin{table}[H]
\centering
	\begin{tabular}{c|cc}
		\toprule
		$\tilde{\sigma_0}$ & Black-Scholes & Bachelier \\
		\midrule
		Heston    & $ \sqrt{v_0} $     & $ 100 \cdot \sqrt{v_0} $ \\
		SABR      & $ 0.1 \cdot v_0 $  & $ 10 \cdot v_0 $  \\
		\bottomrule
	\end{tabular}
\end{table}
One or two of these dynamics were used depending on the option type. Namely, Black-Scholes was used for Vanilla and Barrier options, while Bachelier for Vanilla, Asian and Basket options. This choice was made so that analytical valuation formulae existed for a given payoff type. In case of barrier options this created some complication, as - by the nature of numerical simulation - the process $m_t$ used for the indication of the barrier breach is actually observed in a discrete time and the impact of this discretization on the option value is not negligible. To account for this effect, while maintaining computational efficiency we used approximate valuation formula described in \cite{BroGla97} (a formula for options with adjusted barrier level and continuous-time observations)

\subsection{Test results}
Calculation results are presented below in four tables, organized as follows:
\begin{itemize}
	\item Data in the Tables 1 and 2 refer to the estimates of option values. Tables 3 and 4 contain analogous data for option deltas. In case of basket options, the delta corresponding to the underlying with the highest volatility is presented. 
	\item Table 1 and 3 contains a comparison of option values (resp. deltas) estimated using crude Monte Carlo and our method. The last column presents a Z-score, calculated as a difference between both estimators and divided by its standard deviation.
	 Additionally, in Table 1 results of our method are presented in 2 versions, using Riemann and quadrature integration formula. 
	\item Table 2 and 4 compare the standard errors of aforementioned estimates. The last column contains a variance reduction factor, calculated as a square of standard errors ratio (crude MC vs our).
\end{itemize}
Based on these results the following observations can be made:
\begin{itemize}
	\item In the vast majority of cases an absolute values of Z-score does not exceed 1.96. Hence, in all these cases the difference between results calculated with our method and with crude Monte Carlo can be attributed to purely random estimation error at the confidence level of 99\%. This constitutes an empirical evidence of correctness / accuracy of our method
	\item Option values calculated with our method using Riemann and quadrature integration techniques are almost identical (the difference being orders of magnitude smaller than the estimation error).
	\item The variance reduction factor varies significantly, depending on the option type and parameters. However, typically it is as large as several dozens for option values and above 10 for option deltas. In general it seems that the worst results are obtained for long-dated out-of-the-money options with higher volatilities (here in SABR).
\end{itemize}

\begin{rem}
	In the Tables 3 and 4 deltas for barrier options were omitted. This is because its calculation using algorithmic differentiation for discontinuous payoffs is not straightforward and requires additional artificial smoothing of the indicator function. We managed to do it for our method, but encountered more difficulties for the raw MC payoff and decided not to invest much time on it. Instead, we calculated raw MC delta by bumping \& revaluation and found them in line with our method. However, we cannot present the comparison of monte carlo error of both methods, hence we omitted this case. 
\end{rem}

\begin{table}[H]
   \centering
	\caption{ \label{tabPVavg} Estimates of expected payoff for various option types and simplified dynamics. Comparison of Crude MC average over 1,000,000 simulations with our Riemann sum and Gauss-Legendre quadrature methods, using 5,000 simulations. }
	\vspace{3mm}
	\begin{adjustbox}{width=\textwidth}
	\begin{tabular}{lll|cr|rrrr}
		\toprule
		Dynamics & Payoff & Simplified & Maturity & Strike & Crude MC & Riemann & Legendre & Z-score \\
		\midrule
		Heston & Vanilla & Black-Scholes & 1Y & 105 & 4.1317 & 4.1555 & 4.1579 & 1.56 \\
		Heston & Vanilla & Bachelier & 1Y & 105 & 4.1293 & 4.1499 & 4.1522 & 1.36 \\
		Heston & Vanilla & Black-Scholes & 1Y & 112 & 1.6215 & 1.6323 & 1.6356 & 0.97 \\
		Heston & Vanilla & Bachelier & 1Y & 112 & 1.6225 & 1.6336 & 1.6327 & 0.68 \\
		Heston & Vanilla & Black-Scholes & 5Y & 128 & 11.5222 & 11.5450 & 11.5396 & 0.57 \\
		Heston & Vanilla & Bachelier & 5Y & 128 & 11.5223 & 11.5513 & 11.5663 & 1.43 \\
		Heston & Vanilla & Black-Scholes & 5Y & 149 & 4.5438 & 4.5421 & 4.5373 & -0.25 \\
		Heston & Vanilla & Bachelier & 5Y & 149 & 4.5353 & 4.5607 & 4.5635 & 0.78 \\
		\midrule
		Heston & Asian & Bachelier & 1Y & 103 & 2.8355 & 2.8530 & 2.8523 & 1.47 \\
		Heston & Asian & Bachelier & 1Y & 106 & 1.5900 & 1.6049 & 1.6043 & 1.34 \\
		Heston & Asian & Bachelier & 5Y & 114 & 6.3486 & 6.3651 & 6.3683 & 1.06 \\
		Heston & Asian & Bachelier & 5Y & 129 & 1.7202 & 1.7308 & 1.7331 & 0.78 \\
		\midrule
		Heston & Barrier & Black-Scholes & 1Y & 105 & 3.6048 & 3.6213 & 3.6250 & 1.59 \\
		Heston & Barrier & Black-Scholes & 1Y & 112 & 1.4587 & 1.4709 & 1.4747 & 1.24 \\
		Heston & Barrier & Black-Scholes & 5Y & 128 & 8.0206 & 7.9970 & 7.9949 & -1.04 \\
		Heston & Barrier & Black-Scholes & 5Y & 149 & 3.4004 & 3.3735 & 3.3650 & -1.64 \\
		\midrule
		Heston & Basket & Bachelier & 1Y & 105 & 2.7273 & 2.7336 & 2.7336 & 1.15 \\
		Heston & Basket & Bachelier & 1Y & 112 & 0.5639 & 0.5678 & 0.5679 & 1.23 \\
		Heston & Basket & Bachelier & 5Y & 128 & 7.3833 & 7.3836 & 7.3849 & 0.11 \\
		Heston & Basket & Bachelier & 5Y & 149 & 1.4158 & 1.4050 & 1.4057 & -1.03 \\
		\midrule
		SABR & Vanilla & Black-Scholes & 1Y & 100 & 10.0623 & 10.1307 & 10.1308 & 1.60 \\
		SABR & Vanilla & Bachelier & 1Y & 100 & 10.0524 & 10.1398 & 10.1419 & 2.08 \\
		SABR & Vanilla & Black-Scholes & 1Y & 118 & 3.9621 & 3.9872 & 3.9870 & 0.66 \\
		SABR & Vanilla & Bachelier & 1Y & 118 & 3.9673 & 3.9877 & 3.9883 & 0.55 \\
		SABR & Vanilla & Black-Scholes & 5Y & 100 & 22.9696 & 23.1963 & 23.1915 & 0.99 \\
		SABR & Vanilla & Bachelier & 5Y & 100 & 22.9937 & 23.2097 & 23.2169 & 0.99 \\
		SABR & Vanilla & Black-Scholes & 5Y & 146 & 9.9313 & 9.9381 & 9.9620 & 0.14 \\
		SABR & Vanilla & Bachelier & 5Y & 146 & 9.9614 & 9.9435 & 9.9269 & -0.15 \\
		\midrule
		SABR & Asian & Bachelier & 1Y & 100 & 6.8722 & 6.9142 & 6.9141 & 1.88 \\
		SABR & Asian & Bachelier & 1Y & 108 & 3.7524 & 3.7936 & 3.7924 & 1.93 \\
		SABR & Asian & Bachelier & 5Y & 100 & 13.7470 & 13.7576 & 13.7672 & 0.23 \\
		SABR & Asian & Bachelier & 5Y & 135 & 2.8534 & 2.8121 & 2.8073 & -0.56 \\
		\midrule
		SABR & Barrier & Black-Scholes & 1Y & 100 & 7.7704 & 7.8016 & 7.8016 & 1.23 \\
		SABR & Barrier & Black-Scholes & 1Y & 118 & 3.2642 & 3.2857 & 3.2857 & 0.72 \\
		SABR & Barrier & Black-Scholes & 5Y & 100 & 15.8369 & 15.8120 & 15.7985 & -0.44 \\
		SABR & Barrier & Black-Scholes & 5Y & 146 & 6.5874 & 6.5224 & 6.5333 & -0.41 \\
		\midrule
		SABR & Basket & Bachelier & 1Y & 100 & 7.3148 & 7.3398 & 7.3399 & 1.68 \\
		SABR & Basket & Bachelier & 1Y & 118 & 1.8283 & 1.8357 & 1.8358 & 0.71 \\
		SABR & Basket & Bachelier & 5Y & 100 & 16.1575 & 16.1202 & 16.1218 & -0.61 \\
		SABR & Basket & Bachelier & 5Y & 146 & 4.2730 & 4.2318 & 4.2412 & -0.51 \\
		\midrule
		SABR & Rainbow & Black-Scholes & 1Y & 100 & 20.3989 & 20.4451 & 20.4297 & 0.54 \\
		SABR & Rainbow & Black-Scholes & 1Y & 118 & 9.5983 & 9.6501 & 9.6216 & 0.43 \\				
		\bottomrule
	\end{tabular}
	\end{adjustbox}
\end{table}

\newpage	
\begin{table}[H]
	\centering
	\caption{ \label{tabPVstd} Standard error of expected payoff estimates and a variance reduction ratio. Comparison of Crude MC and our Gauss-Legendre quadrature method for 5,000 simulations. }
	\vspace{3mm}
	\begin{adjustbox}{width=\textwidth}
	\begin{tabular}{lll|cr|rrr}	
		\toprule
		Dynamics & Payoff & Simplified & Maturity & Strike & Crude MC & Legendre & Variance Ratio \\
		\midrule
		Heston & Vanilla & Black-Scholes & 1Y & 105 & 0.0914 & 0.0154 & 35.1 \\
		Heston & Vanilla & Bachelier & 1Y & 105 & 0.0915 & 0.0155 & 34.7 \\
		Heston & Vanilla & Black-Scholes & 1Y & 112 & 0.0602 & 0.0139 & 18.6 \\
		Heston & Vanilla & Bachelier & 1Y & 112 & 0.0602 & 0.0143 & 17.7 \\
		Heston & Vanilla & Black-Scholes & 5Y & 128 & 0.2700 & 0.0240 & 126.4 \\
		Heston & Vanilla & Bachelier & 5Y & 128 & 0.2697 & 0.0240 & 126.1 \\
		Heston & Vanilla & Black-Scholes & 5Y & 149 & 0.1776 & 0.0229 & 60.2 \\
		Heston & Vanilla & Bachelier & 5Y & 149 & 0.1772 & 0.0337 & 27.6 \\
		\midrule
		Heston & Asian & Bachelier & 1Y & 103 & 0.0613 & 0.0105 & 33.8 \\
		Heston & Asian & Bachelier & 1Y & 106 & 0.0475 & 0.0101 & 22.0 \\
		Heston & Asian & Bachelier & 5Y & 114 & 0.1430 & 0.0155 & 84.6 \\
		Heston & Asian & Bachelier & 5Y & 129 & 0.0781 & 0.0156 & 25.0 \\
		\midrule
		Heston & Barrier & Black-Scholes & 1Y & 105 & 0.0890 & 0.0110 & 65.8 \\
		Heston & Barrier & Black-Scholes & 1Y & 112 & 0.0578 & 0.0123 & 22.1 \\
		Heston & Barrier & Black-Scholes & 5Y & 128 & 0.2482 & 0.0175 & 200.2 \\
		Heston & Barrier & Black-Scholes & 5Y & 149 & 0.1596 & 0.0184 & 74.9 \\
		\midrule
		Heston & Basket & Bachelier & 1Y & 105 & 0.0573 & 0.0036 & 249.9 \\
		Heston & Basket & Bachelier & 1Y & 112 & 0.0261 & 0.0026 & 97.2 \\
		Heston & Basket & Bachelier & 5Y & 128 & 0.1626 & 0.0068 & 579.7 \\
		Heston & Basket & Bachelier & 5Y & 149 & 0.0723 & 0.0084 & 74.4 \\
		\midrule
		SABR & Vanilla & Black-Scholes & 1Y & 100 & 0.2342 & 0.0395 & 35.2 \\
		SABR & Vanilla & Bachelier & 1Y & 100 & 0.2343 & 0.0398 & 34.6 \\
		SABR & Vanilla & Black-Scholes & 1Y & 118 & 0.1557 & 0.0358 & 18.9 \\
		SABR & Vanilla & Bachelier & 1Y & 118 & 0.1557 & 0.0366 & 18.1 \\
		SABR & Vanilla & Black-Scholes & 5Y & 100 & 0.7679 & 0.2172 & 12.5 \\
		SABR & Vanilla & Bachelier & 5Y & 100 & 0.7712 & 0.2197 & 12.3 \\
		SABR & Vanilla & Black-Scholes & 5Y & 146 & 0.6262 & 0.2129 & 8.7 \\
		SABR & Vanilla & Bachelier & 5Y & 146 & 0.6299 & 0.2316 & 7.4 \\
		\midrule
		SABR & Asian & Bachelier & 1Y & 100 & 0.1536 & 0.0194 & 62.5 \\
		SABR & Asian & Bachelier & 1Y & 108 & 0.1166 & 0.0191 & 37.4 \\
		SABR & Asian & Bachelier & 5Y & 100 & 0.3658 & 0.0832 & 19.3 \\
		SABR & Asian & Bachelier & 5Y & 135 & 0.2072 & 0.0817 & 6.4 \\
		\midrule
		SABR & Barrier & Black-Scholes & 1Y & 100 & 0.2212 & 0.0199 & 123.9 \\
		SABR & Barrier & Black-Scholes & 1Y & 118 & 0.1443 & 0.0282 & 26.2 \\
		SABR & Barrier & Black-Scholes & 5Y & 100 & 0.6022 & 0.0760 & 62.7 \\
		SABR & Barrier & Black-Scholes & 5Y & 146 & 0.4593 & 0.1278 & 12.9 \\
		\midrule
		SABR & Basket & Bachelier & 1Y & 100 & 0.1618 & 0.0096 & 285.1 \\
		SABR & Basket & Bachelier & 1Y & 118 & 0.0822 & 0.0089 & 84.5 \\
		SABR & Basket & Bachelier & 5Y & 100 & 0.4215 & 0.0498 & 71.5 \\
		SABR & Basket & Bachelier & 5Y & 146 & 0.2452 & 0.0602 & 16.6 \\
		\midrule
		SABR & Rainbow & Black-Scholes & 1Y & 100 & 0.3154 & 0.0546 & 33.4 \\
		SABR & Rainbow & Black-Scholes & 1Y & 118 & 0.2442 & 0.0515 & 22.5 \\				
		\bottomrule
	\end{tabular}
	\end{adjustbox}
\end{table}

\newpage	
\begin{table}[H]
	\centering
	\caption{ \label{tabDavg} Estimates of forward delta for various option types and simplified dynamics. Comparison of Crude MC average over 1,000,000 simulations with our Riemann sum and Gauss-Legendre quadrature methods, using 5,000 simulations. }
	\vspace{3mm}
	\begin{adjustbox}{width=\textwidth}
	\begin{tabular}{lll|cr|rrr}	
		\toprule
		Dynamics & Payoff & Simplified & Maturity & Strike & Crude MC & Legendre & Z-score \\
		\midrule
		Heston & Vanilla & Black-Scholes & 1Y & 105 & 0.5604 & 0.5619 & 0.64 \\
		Heston & Vanilla & Bachelier & 1Y & 105 & 0.5599 & 0.5645 & 1.90 \\
		Heston & Vanilla & Black-Scholes & 1Y & 112 & 0.2813 & 0.2840 & 1.28 \\
		Heston & Vanilla & Bachelier & 1Y & 112 & 0.2813 & 0.2864 & 2.06 \\
		Heston & Vanilla & Black-Scholes & 5Y & 128 & 0.7121 & 0.7142 & 0.90 \\		
		Heston & Vanilla & Bachelier & 5Y & 128 & 0.7121 & 0.7082 & -1.15 \\
		Heston & Vanilla & Black-Scholes & 5Y & 149 & 0.3684 & 0.3739 & 2.21 \\
		Heston & Vanilla & Bachelier & 5Y & 149 & 0.3687 & 0.3789 & 2.10 \\
		\midrule
		Heston & Asian & Bachelier & 1Y & 103 & 0.5485 & 0.5510 & 1.16 \\
		Heston & Asian & Bachelier & 1Y & 106 & 0.3639 & 0.3664 & 1.17 \\
		Heston & Asian & Bachelier & 5Y & 114 & 0.6129 & 0.6153 & 0.96 \\
		Heston & Asian & Bachelier & 5Y & 129 & 0.2314 & 0.2313 & -0.03 \\
		\midrule
		Heston & Basket & Bachelier & 1Y & 105 & 0.0556 & 0.0180 & 0.55 \\
		Heston & Basket & Bachelier & 1Y & 112 & 0.0177 & 0.0180 & \textcolor{black}{2.95} \\
		Heston & Basket & Bachelier & 5Y & 128 & 0.0001 & 0.0001 & 0.00 \\
		Heston & Basket & Bachelier & 5Y & 149 & 0.0042 & 0.0001 & -0.18 \\
		\midrule
		SABR & Vanilla & Black-Scholes & 1Y & 100 & 0.5251 & 0.5257 & 0.23 \\
		SABR & Vanilla & Bachelier & 1Y & 100 & 0.5245 & 0.5261 & 0.70 \\
		SABR & Vanilla & Black-Scholes & 1Y & 118 & 0.2588 & 0.2601 & 0.59 \\		
		SABR & Vanilla & Bachelier & 1Y & 118 & 0.2588 & 0.2634 & 1.88 \\
		SABR & Vanilla & Black-Scholes & 5Y & 100 & 0.5595 & 0.5663 & 0.79 \\
		SABR & Vanilla & Bachelier & 5Y & 100 & 0.5593 & 0.5530 & -0.82 \\
		SABR & Vanilla & Black-Scholes & 5Y & 146 & 0.2489 & 0.2575 & 1.68 \\
		SABR & Vanilla & Bachelier & 5Y & 146 & 0.2490 & 0.2450 & -0.36 \\
		\midrule
		SABR & Asian & Bachelier & 1Y & 100 & 0.5152 & 0.5151 & -0.10 \\
		SABR & Asian & Bachelier & 1Y & 108 & 0.3349 & 0.3355 & 0.28 \\
		SABR & Asian & Bachelier & 5Y & 100 & 0.5271 & 0.5238 & -0.72 \\
		SABR & Asian & Bachelier & 5Y & 135 & 0.1246 & 0.1201 & -0.60 \\
		\midrule
		SABR & Basket & Bachelier & 1Y & 100 & 0.0530 & 0.0531 & 1.37 \\
		SABR & Basket & Bachelier & 1Y & 118 & 0.0193 & 0.0195 & 1.88 \\
		SABR & Basket & Bachelier & 5Y & 100 & 0.0558 & 0.0559 & 0.47 \\
		SABR & Basket & Bachelier & 5Y & 146 & 0.0197 & 0.0196 & -0.08 \\
		\bottomrule
	\end{tabular}
	\end{adjustbox}
\end{table}

\newpage	
\begin{table}[H]
	\centering
	\caption{ \label{tabPVavg} Standard error of delta estimates and a variance reduction ratio. Comparison of the Crude MC and our Gauss-Legendre quadrature method for 5,000 simulations. }
	\vspace{3mm}
	\begin{adjustbox}{width=\textwidth}
	\begin{tabular}{lll|cr|rrr}	
		\toprule
		Dynamics & Payoff & Simplified & Maturity & Strike & Crude MC & Legendre & Variance Ratio \\
		\midrule
		Heston & Vanilla & Black-Scholes & 1Y & 105 & 0.0080 & 0.0022 & 13.0 \\
		Heston & Vanilla & Bachelier & 1Y & 105 & 0.0080 & 0.0024 & 11.6 \\
		Heston & Vanilla & Black-Scholes & 1Y & 112 & 0.0072 & 0.0021 & 11.8 \\
		Heston & Vanilla & Bachelier & 1Y & 112 & 0.0072 & 0.0025 & 8.5 \\
		Heston & Vanilla & Black-Scholes & 5Y & 128 & 0.0110 & 0.0022 & 25.3 \\
		Heston & Vanilla & Bachelier & 5Y & 128 & 0.0110 & 0.0033 & 10.9 \\
		Heston & Vanilla & Black-Scholes & 5Y & 149 & 0.0100 & 0.0024 & 17.2 \\
		Heston & Vanilla & Bachelier & 5Y & 149 & 0.0100 & 0.0048 & 4.4 \\
		\midrule
		Heston & Asian & Bachelier & 1Y & 103 & 0.0077 & 0.0021 & 13.9 \\
		Heston & Asian & Bachelier & 1Y & 106 & 0.0074 & 0.0021 & 12.0 \\
		Heston & Asian & Bachelier & 5Y & 114 & 0.0091 & 0.0025 & 13.0 \\
		Heston & Asian & Bachelier & 5Y & 129 & 0.0074 & 0.0032 & 5.4 \\
		\midrule
		Heston & Basket & Bachelier & 1Y & 105 & 0.0008 & 0.0001 & 66.1 \\
		Heston & Basket & Bachelier & 1Y & 112 & 0.0006 & 0.0001 & 61.5 \\
		Heston & Basket & Bachelier & 5Y & 128 & 0.0594 & 0.0654 & \textcolor{black}{0.8} \\
		Heston & Basket & Bachelier & 5Y & 149 & 0.2348 & 0.0166 & 200.0 \\
		\midrule
		SABR & Vanilla & Black-Scholes & 1Y & 100 & 0.0079 & 0.0024 & 11.0 \\
		SABR & Vanilla & Bachelier & 1Y & 100 & 0.0079 & 0.0022 & 12.8 \\
		SABR & Vanilla & Black-Scholes & 1Y & 118 & 0.0069 & 0.0022 & 10.3 \\
		SABR & Vanilla & Bachelier & 1Y & 118 & 0.0069 & 0.0024 & 8.2 \\
		SABR & Vanilla & Black-Scholes & 5Y & 100 & 0.0096 & 0.0086 & \textcolor{black}{1.2} \\
		SABR & Vanilla & Bachelier & 5Y & 100 & 0.0098 & 0.0077 & \textcolor{black}{1.6} \\
		SABR & Vanilla & Black-Scholes & 5Y & 146 & 0.0084 & 0.0051 & \textcolor{black}{2.7} \\
		SABR & Vanilla & Bachelier & 5Y & 146 & 0.0087 & 0.0111 & \textcolor{black}{0.6} \\
		\midrule
		SABR & Asian & Bachelier & 1Y & 100 & 0.0076 & 0.0019 & 16.7 \\
		SABR & Asian & Bachelier & 1Y & 108 & 0.0072 & 0.0019 & 13.6 \\
		SABR & Asian & Bachelier & 5Y & 100 & 0.0082 & 0.0046 & \textcolor{black}{3.2} \\
		SABR & Asian & Bachelier & 5Y & 135 & 0.0056 & 0.0075 & \textcolor{black}{0.6} \\
		\midrule
		SABR & Basket & Bachelier & 1Y & 100 & 0.0008 & 0.0001 & 95.7 \\
		SABR & Basket & Bachelier & 1Y & 118 & 0.0006 & 0.0001 & 44.5 \\
		SABR & Basket & Bachelier & 5Y & 100 & 0.0013 & 0.0003 & 23.0 \\
		SABR & Basket & Bachelier & 5Y & 146 & 0.0010 & 0.0005 & \textcolor{black}{3.8} \\
		\bottomrule
	\end{tabular}
	\end{adjustbox}
\end{table}

\newpage	
\section{Acknowledgement}
The views expressed in this article are those of the authors alone and do not necessarily represent those of the author's employers.
The analyses and conclusions set forth are those of the authors and do not reflect in any way the opinions of any institutions in which they currently or have previously held a position. Any such institutions are not responsible for any statement, conclusions not the positions and theories presented herein. 

\newpage

\printbibliography

\newpage	
\section*{Appendix}
\textbf{Proof of Lemma \ref{lemma}:}
\begin{proof}
The proof can be performed by simply comparing the left and right-hand sides of equation \eqref{eq:lemm}, after some rearrangements. By using the chain rule, the left-hand side can be expressed as
\begin{equation}
\begin{split}
\nabla^{2} (f \circ \tau) (u(x)) &= \sum_{i=1}^R \frac{\partial^2 (f \circ \tau)}{\partial u_i^2} (u(x)) =
\sum_{i=1}^R \frac{\partial}{\partial u_i}\left(\sum_{j=1}^d \frac{\partial f(x)}{\partial x_j} \tilde{c}_j  c_{ji}\right) = \sum_{i=1}^R \sum_{j=1}^d \tilde{c}_j c_{ji} \sum_{k=1}^d \ppder{f(x)}{x_j}{x_k} \tilde{c}_k  c_{ki}\\&=
\sum_{j=1}^d \sum_{k=1}^d \ppder{f(x)}{x_j}{x_k} \tilde{c}_j \tilde{c}_k  \left( \sum_{i=1}^R  c_{ji} c_{ki} \right) =
\sum_{j=1}^d \sum_{k=1}^d \ppder{f(x)}{x_j}{x_k} \tilde{c}_j \tilde{c}_k \rho_{jk}
\end{split}
\end{equation}
From the elementary properties of the trace operator, the right-hand side can be expressed as
\begin{equation}
\begin{split}
\tr \left( A^T \kappa(x) A \right) &= \tr \left( \kappa(x) AA^T \right) = \tr \left(\kappa(x) diag(\tilde{c}) \rho diag(\tilde{c})\right) \\&= 
\sum_{j=1}^d \sum_{k=1}^d \ppder{f(x)}{x_j}{x_k} \tilde{c}_j \tilde{c}_k \rho_{kj} = \sum_{j=1}^d \sum_{k=1}^d \ppder{f(x)}{x_j}{x_k} \tilde{c}_j \tilde{c}_k \rho_{jk}
\end{split}
\end{equation}
Consequently, the left and right-hand side have identical expressions, which proves the assertion.
\end{proof}	

\end{document}